\newif\ifdraft  \draftfalse  
\newif\ifpdflatex  \pdflatexfalse 

\documentclass[12pt,a4paper,UKenglish]{article}
\usepackage[
  a4paper
  ,margin=1.25in
  ,foot=0.5in
  ]{geometry}
\usepackage{xspace}
\usepackage{pstricks-add}
\usepackage{scalefnt}
\usepackage{amssymb,amsmath}
\usepackage{amsthm}
\usepackage{amsfonts}
\usepackage{enumitem}\setlist{nolistsep}

\usepackage{mathtools}
\usepackage[mathlines]{lineno}
\usepackage[T1]{fontenc}
\ifx\optionkeymacros\undefined\else \fi
\catcode`\=\active\def{\c{c}}      
\catcode`\'=\active\def'{\c{C}}      

\catcode`\Ï=\active\defÏ{{\oe}}      
\catcode`\Î=\active\defÎ{{\OE}}      

\catcode`\¾=\active\def¾{{\ae}}      
\catcode`\®=\active\def®{{\AE}}      

\catcode`\Š=\active\defŠ{\"a}        
\catcode`\‰=\active\def‰{\^a}        
\catcode`\ˆ=\active\defˆ{\`a}        
\catcode`\‡=\active\def‡{\'a}        
\catcode`\‹=\active\def‹{\~a}        
\catcode`\€=\active\def€{\"A}        
\catcode`\å=\active\defå{\^A}        
\catcode`\Ë=\active\defË{\`A}        
\catcode`\ç=\active\defç{\'A}        
\catcode`\Ì=\active\defÌ{\~A}        

\catcode`\'=\active\def'{\"e}        
\catcode`\=\active\def{\^e}        
\catcode`\=\active\def{\`e}        
\catcode`\Ž=\active\defŽ{\'e}        
\catcode`\è=\active\defè{\"E}        
\catcode`\æ=\active\defæ{\^E}        
\catcode`\é=\active\defé{\`E}        
\catcode`\ƒ=\active\defƒ{\'E}        

\catcode`\•=\active\def•{\"{\i}}     
\catcode`\"=\active\def"{\^{\i}}     
\catcode`\"=\active\def"{\`{\i}}     
\catcode`\'=\active\def'{\'{\i}}     
\catcode`\ì=\active\defì{\"I}        
\catcode`\ë=\active\defë{\^I}        
\catcode`\í=\active\defí{\`I}        
\catcode`\ê=\active\defê{\'I}        

\catcode`\š=\active\defš{\"o}        
\catcode`\™=\active\def™{\^o}        
\catcode`\˜=\active\def˜{\`o}        
\catcode`\—=\active\def—{\'o}        
\catcode`\›=\active\def›{\~o}        
\catcode`\…=\active\def…{\"O}        
\catcode`\ï=\active\defï{\^O}        
\catcode`\ñ=\active\defñ{\`O}        
\catcode`\î=\active\defî{\'O}        
\catcode`\Í=\active\defÍ{\~O}        

\catcode`\Ÿ=\active\defŸ{\"u}        
\catcode`\ž=\active\defž{\^u}        
\catcode`\=\active\def{\`u}        
\catcode`\œ=\active\defœ{\'u}        
\catcode`\†=\active\def†{\"U}        
\catcode`\ó=\active\defó{\^U}        
\catcode`\ô=\active\defô{\`U}        
\catcode`\ò=\active\defò{\'U}        

\catcode`\Ø=\active\defØ{\"y}        
\catcode`\Ù=\active\defÙ{\"Y}        

\catcode`\–=\active\def–{\~n}        
\catcode`\"=\active\def"{\~N}        

\let\optionkeymacros\null

\usepackage{float}
\usepackage{MnSymbol}

\usepackage{algorithmic}
\algsetup{indent=2em,linenodelimiter=.}

\floatstyle{ruled}
\newfloat{procedure}{ht}{pc}
\floatname{procedure}{Procedure}
\usepackage{caption}
\usepackage{subfig}
\usepackage{graphicx}
\usepackage{multirow}
\title{Breadth-first serialisation of trees \\ and rational languages}
\author
{
        Victor Marsault\thanks{Corresponding author, \texttt{victor.marsault@telecom-paristech.fr}}~${}^{,}$\footnotemark~
        \addtocounter{footnote}{-1}
        and 
        Jacques Sakarovitch\thanks{LTCI, CNRS / Telecom ParisTech}~
}
\date{2014\,--\,04\,--\,03}

\newcommand{\BIBINPUTSDIR}{}

\usepackage{ifthen}
\usepackage{etoolbox}
\usepackage{xargs}
\usepackage{slantsc}

\newcommandx{\newtheoremy}[3][2={}]{
  \ifthenelse{\equal{#2}{}}{
    \ifcsmacro{#1}{}{\newtheorem{#1}{#3}}
  }{
    \ifcsmacro{#1}{}{\newtheorem{#1}[#2]{#3}}
  }
}

\newcommand{\thmBlockFont}[1]{#1}
\newtheoremy{thm}{thm}
\newtheoremy{algorithm}[thm]{\thmBlockFont{Algorithm}}
\newtheoremy{corollary}[thm]{\thmBlockFont{Corollary}}
\newtheoremy{conjecture}[thm]{\thmBlockFont{Conjecture}}
\newtheoremy{definition}[thm]{\thmBlockFont{Definition}}
\newtheoremy{example}[thm]{\thmBlockFont{Example}}
\newtheoremy{lemma}[thm]{\thmBlockFont{Lemma}}
\newtheoremy{proposition}[thm]{\thmBlockFont{Proposition}}
\newtheoremy{property}[thm]{\thmBlockFont{Property}}
\newtheoremy{question}[thm]{\thmBlockFont{Question}}
\newtheoremy{remark}[thm]{\thmBlockFont{Remark}}
\newtheoremy{notation}[thm]{\thmBlockFont{Notation}}
\newtheoremy{theorem}[thm]{\thmBlockFont{Theorem}}

\newcommand{\ldefinition}[1]{\label{d.#1}}

\newcommand{\llemma}[1]{\label{l.#1}}
\newcommand{\lproposition}[1]{\label{p.#1}}

\newcommand{\lremark}[1]{\label{r.#1}}

\newcommand{\lsection}[1]{\label{s.#1}}
\newcommand{\ltable}[1]{\label{t.#1}}
\newcommand{\lfigure}[1]{\label{f.#1}}
\newcommand{\ltheorem}[1]{\label{t.#1}}
\newcommand{\lequation}[1]{\label{eq.#1}}

\newcommand{\generalref}[2]{%
  \ifthenelse{\equal{#1}{eq}}%
  {(\ref{#1.#2})}%
  {\ref{#1.#2}}%
}
\newcommand{\generalpageref}[2]{\pageref{#1.#2}}

\makeatletter
\newcommand*{\ralgorithm}{\@ifstar{\generalref{a}}{Algorithm~\ralgorithm*}}
\newcommand*{\palgorithm}{\@ifstar{\generalpageref{a}}{page~\palgorithm*}}

\newcommand*{\rcorollary}{\@ifstar{\generalref{c}}{Corollary~\rcorollary*}}
\newcommand*{\pcorollary}{\@ifstar{\generalpageref{c}}{page~\pcorollary*}}

\newcommand*{\rconjecture}{\@ifstar{\generalref{cj}}{Conjecture~\rconjecture*}}
\newcommand*{\pconjecture}{\@ifstar{\generalpageref{cj}}{page~\pconjecture*}}

\newcommand*{\rdefinition}{\@ifstar{\generalref{d}}{Definition~\rdefinition*}}
\newcommand*{\pdefinition}{\@ifstar{\generalpageref{d}}{page~\pdefinition*}}

\newcommand*{\rexample}{\@ifstar{\generalref{e}}{Example~\rexample*}}
\newcommand*{\pexample}{\@ifstar{\generalpageref{e}}{page~\pexample*}}

\newcommand*{\rlemma}{\@ifstar{\generalref{l}}{Lemma~\rlemma*}}
\newcommand*{\plemma}{\@ifstar{\generalpageref{l}}{page~\plemma*}}

\newcommand*{\rproposition}{\@ifstar{\generalref{p}}{Proposition~\rproposition*}}
\newcommand*{\pproposition}{\@ifstar{\generalpageref{p}}{page~\pproposition*}}

\newcommand*{\rproperty}{\@ifstar{\generalref{pp}}{Proposition~\rproperty*}}
\newcommand*{\pproperty}{\@ifstar{\generalpageref{pp}}{page~\pproperty*}}

\newcommand*{\rprocedure}{\@ifstar{\generalref{pc}}{Procedure~\rprocedure*}}
\newcommand*{\pprocedure}{\@ifstar{\generalpageref{pc}}{page~\pprocedure*}}

\newcommand*{\rremark}{\@ifstar{\generalref{r}}{Remark~\rremark*}}
\newcommand*{\premark}{\@ifstar{\generalpageref{r}}{page~\premark*}}

\newcommand*{\rnotation}{\@ifstar{\generalref{n}}{Notation~\rnotation*}}
\newcommand*{\pnotation}{\@ifstar{\generalpageref{n}}{page~\pnotation*}}

\newcommand*{\rsection}{\@ifstar{\generalref{s}}{Section~\rsection*}}
\newcommand*{\psection}{\@ifstar{\generalpageref{s}}{page~\psection*}}

\newcommand*{\rtable}{\@ifstar{\generalref{t}}{Table~\rtable*}}
\newcommand*{\ptable}{\@ifstar{\generalpageref{t}}{page~\ptable*}}

\newcommand*{\rfigure}{\@ifstar{\generalref{f}}{Figure~\rfigure*}}
\newcommand*{\pfigure}{\@ifstar{\generalpageref{f}}{page~\pfigure*}}

\newcommand*{\requation}{\@ifstar{\generalref{eq}}{Equation~\requation*}}
\newcommand*{\pequation}{\@ifstar{\generalpageref{eq}}{page~\pequation*}}

\newcommand*{\rtheorem}{\@ifstar{\generalref{t}}{Theorem~\rtheorem*}}
\newcommand*{\ptheorem}{\@ifstar{\generalpageref{t}}{page~\ptheorem*}}

\makeatother

\usepackage{atbegshi,picture}
\def\Vhrulefill{\leavevmode\leaders\hrule height 0.7ex depth \dimexpr0.4pt-0.7ex\hfill\kern0pt}

\newcommandx{\wlen}[1]{|#1|}

\newcommandx{\cod}[1]{\langle #1 \rangle}
\newcommandx{\floor}[1]{\lfloor #1 \rfloor}
\newcommandx{\bfloor}[1]{\left\lfloor #1 \right\rfloor}
\newcommandx{\bceil}[1]{\left\lceil #1 \right\rceil}
\newcommandx{\ceil}[1]{\lceil #1 \rceil}

\newcommandx{\newcommandy}[5][1=i,3=0,4={}]{%
  \ifthenelse{\isundefined{#2}}{\newcommandx{#2}[#3][#4]{#5}}{%
      \ifthenelse{\equal{#1}{i}}{}{}%
      \ifthenelse{\equal{#1}{o}}{\renewcommandx{#2}[#3][#4]{#5}}{}%
    }%
}
\newcommandy[o]{\pathx}[2][2={}]{
  \nlb%
  \ifthenelse{\equal{#2}{}}%
    {\xrightarrow{\ #1 \ }}%
    {\underset{#2}{\xrightarrow{\ #1 \ }}}%
  \nlb%
}
\newcommandy[o]{\pathy}[2][2={}]{\nlb\underset{#2}{\xleftarrow{\ #1 \ }}\nlb}

\newcommand{\val}[1]{\widebar{#1}}

\newcommand{\strong}[1]{\textbf{#1}}
\newcommand{\ssc}[1]{\textbf{\textsc{#1}}}

\newcommand{\set}[1]{\{#1\}}

\newcommand{\Z}{\mathbb{Z}}
\newcommand{\N}{\mathbb{N}}

\newcommand{\widebar}{\overline}

\newcommandy[o]{\Cup}{\bigcup}
\newcommand{\nlb}{\nolinebreak}

\renewcommand{\thmBlockFont}[1]{\ssc{#1}}

\newcommand{\vmfiguretodo}[2][]{%
  \begin{figure}[ht!]
    \frame{%
      \begin{minipage}{\linewidth}
        ~\hfill~
        \vspace*{#2cm}
      \end{minipage}
    }
    \ifthenelse{\equal{#1}{}}{}{\caption{#1}}
  \end{figure}
}

\makeatletter
\newcommandx{\newcommandWithStar}[3][1=i]{%
  \newcommandy[#1]{#2}{\protect\@ifstar{\leavevmode\protect\nlb$\protect#3$}{#3}}
}
\makeatother

\renewcommand{\geq}{\geqslant}
\renewcommand{\phi}{\varphi}
\renewcommand{\epsilon}{\varepsilon}

\newcommand{\e}{\text{\quad}}                 
\newcommand{\ee}{\text{\qquad}}               
\newcommand{\eee}{\text{\qquad \qquad}} 
\newsavebox{\InterSymbolSpace}
\savebox{\InterSymbolSpace}{\hspace{0.125em}}
\newsavebox{\SideFormulaSpace}
\savebox{\SideFormulaSpace}{\hspace{0.2em}}
\newcommand{\xmd}{\usebox{\InterSymbolSpace}} 
\newcommand{\eqpnt}{\makebox[0pt][l]{\: .}}
\newcommand{\eqpntvrg}{\makebox[0pt][l]{\: ;}}

\newcommand{\quantsp}{\ee }

\newcommand{\LatinLocution}[1]{{\itshape #1}\xspace}

\newcommand{\cf}{\LatinLocution{cf.}}

\newcommand{\UNmbb}{{\mathchoice
{\hbox{$\textstyle\rm 1\kern-0.2em I$}}%
{\hbox{$\textstyle\rm 1\kern-0.2em I$}}%
{\hbox{$\scriptstyle\rm 1\kern-0.15em I$}}%
{\hbox{$\scriptscriptstyle\rm 1\kern-0.1em I$}}%
}}
\newcommand{\Ac}{\mathcal{A}}
\newcommand{\Bc}{\mathcal{B}}

\newlength{\ArrowDiagSize}
\newlength{\ArrowDiagWidth}
\newpsstyle{SLDiagStyle}%
   {colsep=6ex,rowsep=5ex,nodesep=1ex,npos=.45,%
    arrows=->,linewidth=\ArrowDiagWidth,arrowsize=\ArrowDiagSize,%
        linestyle=solid,linecolor=\ArrowDiagColor}

\newenvironment{SLDiag}%
   {\psset{style=SLDiagStyle}\begin{psmatrix}}%
   {\end{psmatrix}}%
\newcommand{\CDSL}{\begin{SLDiag}}
\newcommand{\CDSLF}{\end{SLDiag}}
\newenvironment{DiagraBig}%
{\psmatrix[colsep=7ex,rowsep=6ex,arrows=->,nodesep=1ex,npos=.45]}%
{\endpsmatrix}
\newcommand{\CDB}{\begin{DiagraBig}}
\newcommand{\CDBF}{\end{DiagraBig}}
\newenvironment{DiagraSmall}%
{\psmatrix[colsep=3ex,rowsep=3ex,arrows=->,nodesep=1ex,npos=.45]}%
{\endpsmatrix}
\newcommand{\CDS}{\begin{DiagraSmall}}
\newcommand{\CDSF}{\end{DiagraSmall}}

\newcommand{\matriceuu}[1]%
    {\begin{pmatrix} #1 \end{pmatrix}}
\newcommand{\matricedd}[4]%
    {\begin{pmatrix} #1 & #2 \\ #3 & #4 \end{pmatrix}}
\newcommand{\vecteurd}[2]%
    {\begin{pmatrix} #1 \\ #2 \end{pmatrix}}
\newcommand{\ligned}[2]%
    {\begin{pmatrix} #1 & #2 \end{pmatrix}}
\newcommand{\matricett}[9]%
    {\begin{pmatrix}  #1 & #2 & #3 \\
                      #4 & #5 & #6 \\
                      #7 & #8 & #9 \end{pmatrix}}
\newcommand{\vecteurt}[3]%
    {\begin{pmatrix} #1 \\ #2 \\ #3 \end{pmatrix}}
\newcommand{\lignet}[3]%
    {\begin{pmatrix} #1 & #2 & #3 \end{pmatrix}}

\newlength{\jsWidthCol}
\newlength{\blocinterligne}
\newlength{\blocinterligned}
\newlength{\temparraycolsep}
\newlength{\longueurbloc}
\newlength{\hauteurbloc}
\newlength{\centragebloc}
\newlength{\longueurblc}
\newlength{\hauteurblc}
\newlength{\centrageblc}
\newcommand{\blocligne}[1]%
    {\framebox[\longueurbloc]{$#1$}}
\newcommand{\blocmatrice}[1]%
    {\framebox[\longueurbloc]{\rule[\centragebloc]{0mm}{\hauteurbloc}$#1$}}
\newcommand{\blocvecteur}[1]%
    {\framebox{\rule[\centragebloc]{0mm}{\hauteurbloc}$#1$}}
\newcommand{\blcligne}[1]%
    {\framebox[\longueurblc]{$#1$}}
\newcommand{\blcmatrice}[1]%
    {\framebox[\longueurblc]{\rule[\centrageblc]{0mm}{\hauteurblc}$#1$}}
\newcommand{\blcvecteur}[1]%
    {\framebox{\rule[\centrageblc]{0mm}{\hauteurblc}$#1$}}
\newcommand{\matriceddblvs}[4]
   {\setlength{\temparraycolsep}{\arraycolsep}%
    \setlength{\arraycolsep}{1.3pt}%
        \left (%
    \begin{array}{cc}%
                #1  & \blcligne{#2} \\
            \blcvecteur{#3} & \blcmatrice{#4}
        \end{array}%
        \right )%
    \setlength{\arraycolsep}{\temparraycolsep}%
   }%
\newcommand{\vecteurdblvs}[2]%
   {\setlength{\temparraycolsep}{\arraycolsep}%
    \setlength{\arraycolsep}{1.5pt}%
        \left (%
    \begin{array}{c}%
                #1  \\
            \blcvecteur{#2}
        \end{array}%
        \right )%
    \setlength{\arraycolsep}{\temparraycolsep}%
   }%
\newcommand{\lignedblvs}[2]%
   {\setlength{\temparraycolsep}{\arraycolsep}%
    \setlength{\arraycolsep}{1.5pt}%
        \left (%
    \begin{array}{cc}%
                #1  & \blcligne{#2}
        \end{array}%
        \right )%
    \setlength{\arraycolsep}{\temparraycolsep}%
   }%
\newcommand{\matricettblvs}[9]
   {\setlength{\temparraycolsep}{\arraycolsep}%
    \setlength{\arraycolsep}{1.5pt}%
        \left (%
    \begin{array}{ccc}%
                #1  & \blcligne{#2} & #3\\
            \blcvecteur{#4} & \blcmatrice{#5} & \blcvecteur{#6}\\
                #7  & \blcligne{#8} & #9\\
        \end{array}%
        \right )%
    \setlength{\arraycolsep}{\temparraycolsep}%
   }%
\newcommand{\vecteurtblvs}[3]%
   {\setlength{\temparraycolsep}{\arraycolsep}%
    \setlength{\arraycolsep}{1.5pt}%
        \left (%
    \begin{array}{c}%
                #1  \\
            \blcvecteur{#2}\\
                #3
        \end{array}%
        \right )%
    \setlength{\arraycolsep}{\temparraycolsep}%
   }%
\newcommand{\lignetblvs}[3]%
   {\setlength{\temparraycolsep}{\arraycolsep}%
    \setlength{\arraycolsep}{1.5pt}%
        \left (%
    \begin{array}{ccc}%
                #1  & \blcligne{#2} & #3
        \end{array}%
        \right )%
    \setlength{\arraycolsep}{\temparraycolsep}%
   }%
\newcommand{\matricettblblvs}[9]
   {\setlength{\temparraycolsep}{\arraycolsep}%
    \setlength{\arraycolsep}{1.5pt}%
        \left (%
    \begin{array}{ccc}%
                #1  & \blcligne{#2} & \blcligne{#3}\\
            \blcvecteur{#4} & \blcmatrice{#5} & \blcmatrice{#6}\\
                \blcvecteur{#7}  & \blcmatrice{#8} & \blcmatrice{#9}\\
        \end{array}%
        \right )%
    \setlength{\arraycolsep}{\temparraycolsep}%
   }%
\newcommand{\vecteurtblblvs}[3]%
   {\setlength{\temparraycolsep}{\arraycolsep}%
    \setlength{\arraycolsep}{1.5pt}%
        \left (%
    \begin{array}{c}%
                #1  \\
            \blcvecteur{#2}\\
                \blcvecteur{#3}
        \end{array}%
        \right )%
    \setlength{\arraycolsep}{\temparraycolsep}%
   }%
\newcommand{\lignetblblvs}[3]%
   {\setlength{\temparraycolsep}{\arraycolsep}%
    \setlength{\arraycolsep}{1.5pt}%
        \left (%
    \begin{array}{ccc}%
                #1  & \blcligne{#2} & \blcligne{#3}
        \end{array}%
        \right )%
    \setlength{\arraycolsep}{\temparraycolsep}%
   }%
\newlength{\DefiTest}
\newlength{\DefiHeightu}\newlength{\DefiHeightd}%
\newlength{\DefiDepthu}\newlength{\DefiDepthd}%
\newcommand{\Defi}[2]%
    {%
     \settoheight{\DefiHeightu}{${\displaystyle #1}$}%
     \settodepth{\DefiDepthu}{${\displaystyle #1}$}%
     \addtolength{\DefiHeightu}{\DefiDepthu}%
     \settoheight{\DefiHeightd}{${\displaystyle #2}$}%
     \settodepth{\DefiDepthd}{${\displaystyle #2}$}%
     \addtolength{\DefiHeightd}{\DefiDepthd}%
     \left\{#1%
     \rule[-\DefiDepthd]{\DefiTest}{\DefiHeightd}%
     \xmd\right|%
     \left.%
     \rule[-\DefiDepthu]{\DefiTest}{\DefiHeightu}%
      #2\right\}%
     }
\newlength{\ColoText}
\newlength{\ColoFigu}
\newlength{\TextFiguSpace}
\newlength{\parindenttemp} 
\newlength{\parskiptemp} 
\newlength{\fboxseptemp} 
\newcommand{\TFBoxing}{}
\newcommand{\TFVertAlig}{}
\newcommand{\LeftLarg}{}
\renewcommand{\LeftLarg}{.66}
\ifdraft\renewcommand{\TFBoxing}{\fbox}\fi
\newcommand{\TxtFg}[3]%
   {%
    \setlength{\ColoText}{#1\textwidth}%
    \setlength{\ColoFigu}{\textwidth}%
    \addtolength{\ColoFigu}{-\ColoText}%
    \addtolength{\ColoText}{-.5\TextFiguSpace}%
    \addtolength{\ColoFigu}{-.5\TextFiguSpace}%
    \ifdraft\setlength{\fboxsep}{0pt}\fi
    \noi
    \TFBoxing{%
       \begin{minipage}[\TFVertAlig]{\ColoText}%
          \setlength{\parindent}{\parindenttemp}%
          \setlength{\parskip}{\parskiptemp}%
          \par\vspace*{0mm}
             #2
       \end{minipage}%
             }%
    \hspace*{\TextFiguSpace}%
    \TFBoxing{%
       \begin{minipage}[\TFVertAlig]{\ColoFigu}%
          \par\vspace*{0mm}%
             #3%
       \end{minipage}%
             }%
    \ifdraft\setlength{\fboxsep}{\fboxseptemp}\fi
   }%
\newcommand{\TextFigu}[3][\LeftLarg]%
   {\renewcommand{\TFVertAlig}{t}\TxtFg{#1}{#2}{#3}}
\newcommand{\TextFiguC}[3][\LeftLarg]%
   {\renewcommand{\TFVertAlig}{c}\TxtFg{#1}{#2}{#3}}
\newcommand{\TextFiguX}[3][\LeftLarg]
   {%
    \setlength{\ColoText}{#1\textwidth}%
    \setlength{\ColoFigu}{\textwidth}%
    \addtolength{\ColoFigu}{-\ColoText}%
    \addtolength{\ColoText}{-.5\TextFiguSpace}%
    \addtolength{\ColoFigu}{-.5\TextFiguSpace}%
    \addtolength{\ColoFigu}{\ETAExtendedLineWidth}
    \ifdraft\setlength{\fboxsep}{0pt}\fi
    \noi
    \ifodd\value{page}%
       \TFBoxing{%
          \begin{minipage}[t]{\ColoText}%
             \RstBLS
             \setlength{\parindent}{\parindenttemp}%
             \setlength{\parskip}{\parskiptemp}%
             \par\vspace*{0mm}
                #2
          \end{minipage}%
                }%
       \hspace*{\TextFiguSpace}%
       \TFBoxing{%
          \begin{minipage}[t]{\ColoFigu}%
             \par\vspace*{0mm}%
                #3%
          \end{minipage}%
                }%
    \else
       \hspace*{-\ETAExtendedLineWidth}
       \TFBoxing{%
          \begin{minipage}[t]{\ColoFigu}%
             \par\vspace*{0mm}%
                #3%
          \end{minipage}%
                }%
       \hspace*{\TextFiguSpace}%
       \TFBoxing{%
          \begin{minipage}[t]{\ColoText}%
             \RstBLS
             \setlength{\parindent}{\parindenttemp}%
             \setlength{\parskip}{\parskiptemp}%
             \par\vspace*{0mm}
                #2
          \end{minipage}%
                }%
    \fi%
    \ifdraft\setlength{\fboxsep}{\fboxseptemp}\fi
   }

\newcommand{\NoteEnMarge}[1]%
   {%
    \marginpar[\begin{flushright}%
               {\sl {\scriptsize #1}}%
               \end{flushright}]%
              {\begin{flushleft}%
               {\sl {\scriptsize #1}}%
               \end{flushleft}}%
	}%
\newcommand{\Axio}[1]%
   {\pointn #1\hspace*{.1em}\jspointtiret\hspace*{.4em}\ignorespaces}

\newcommandy[o]{\path}{}
\newcommand{\ExtnF}[1]%
   {\overset{{\scriptscriptstyle \pmb{\smile}}}{#1}}

\newcommand{\DiffF}[1]%
   {\overset{{\scriptscriptstyle \pmb{\lor}}}{#1}}

\newcommand{\LocaF}[1]%
   {\overset{{\scriptscriptstyle \leftrightarrow}}{#1}}

\newcommand{\jsDist}[2][{}]%
   {\operatorname{\mathbf{d}_{#1}}\left(#2\right)}

\renewcommand{\lim}{{\operatornamewithlimits{\mathsf{lim}}}}

\newcommand{\SerSAnMon}[2]%
    {#1 \langle \! \langle  #2  \rangle \! \rangle }
\newcommand{\SerSAnMonD}[2]%
    {\left[#1\right] \langle \! \langle  #2  \rangle \! \rangle }
\newcommand{\SerMon}[1]%
    {\!\langle \! \langle  #1  \rangle \! \rangle }
\newcommand{\PolSAnMon}[2]%
    {{#1 \langle  #2 \rangle }}
\newcommand{\PolMon}[1]%
    {{\!\langle  #1 \rangle }}

\newcommand{\jsStar}[1]{{{#1}^{*}}}
\newcommand{\Ae}{\jsStar{A}}

\newcommand{\jsPlus}[1]{{{#1}^{+}}}
\newcommand{\Ap}{\jsPlus{A}}

\newcommand{\iotaK}{\iota_{\ShiftInd{K}}}

\newcommand{\compos}{\ccdot }

\newcommand{\phiikpsi}%
{{\varphi ^{-1}\! \compos        \iotaK \! \compos \! \psi }}
\newcommand{\phiiotpsi}[1]%
{{\varphi ^{-1}\! \compos        \iota _{\ShiftInd{#1}} \! \compos \! \psi }}
\newcommand{\phiintkpsi}[1]%
{{(#1\varphi ^{-1}\! \cap K) \psi }}

\newcommand{\jsless}
   {\mathrel{\leqslant_{\!\!\!\!\scriptscriptstyle{/}}}}
\newcommand{\jsgrea}
   {\mathrel{\geqslant_{\!\!\!\!\scriptscriptstyle{\backslash}}}}

\newcommand{\lexiconeq}
   {\preccurlyeq_{\!\!\!\!\!\scalebox{1.8 1}{\scriptscriptstyle{\pmb{/}}}}}

\newcommand{\jsAutUn}[1]%
   {\mbox{$\left\langle \thinspace #1 \thinspace \right\rangle $}}
\newcommand{\aut}[1]{\jsAutUn{#1}} 

\newcommand{\ShiftInd}[1]{\raisebox{-0.3ex}{$\scriptstyle{#1}$}}

\newcommand{\actb}{\mathbin{\raisebox{0.2ex}%
                        {${\scriptscriptstyle \circ} $}}}
\newcommand{\ccdot}{\actb} 

\newlength{\vbh}\newlength{\vbd}\newlength{\vbt}%
\newcommand{\CompAuto}[1]%
    {%
     \settodepth{\vbd}{\mbox{$\displaystyle{#1\strut}$}}%
     \settoheight{\vbh}{\mbox{$\displaystyle{#1\strut}$}}%
     \setlength{\vbt}{\vbh}\addtolength{\vbt}{\vbd}%
     {}%
     \psline[linewidth=0.8pt]{c-c}(0,-.65\vbd)(0,.9\vbh)%
     \hspace*{0.7pt}%
     {#1}%
     \kern0.8pt%
     \psline[linewidth=0.8pt]{c-c}(0,-.65\vbd)(0,.9\vbh)%
     }%

\newcommand{\bornedeuxlignes}[2]%
{\mbox{$
\begin{array}{c}{\scriptstyle #1}\\ {\scriptstyle #2} \end{array}
       $}}

\renewcommand{\path}[1]{\xrightarrow{\ #1 \ }} 

\newcommand{\ExpDer}[2][a]%
    {\operatorname{\frac{\partial}{\partial \mbox{$#1$}}}#2}
\newcommand{\ExpDerP}[2][a]%
    {\operatorname{\frac{\partial}{\partial\mbox{$#1$}}}\left(#2\right)}
\newcommand{\ExpDerr}[2][a]%
    {\operatorname{\frac{\partial_{\mathrm{R}}}{\partial \mbox{$#1$}}}#2}
\newcommand{\ExpDerB}[2][a]%
   {\operatorname{\frac{\partial_\mathsf{b}}{\partial \mbox{$#1$}}}#2}
\newcommand{\ExpDerBP}[2][a]%
   {\operatorname{\frac{\partial_\mathsf{b}}{\partial \mbox{$#1$}}}\left(#2\right)}

\newcommand{\customthmname}{Theorem}

\usepackage{stmaryrd}

\newcommand{\THFracs}[2]{\frac{#1}{#2}}%
 
\newcommand{\pqs}{\THFracs{p}{q}}

\newcommand{\Indpq}[1]{#1_{\pqs}}%

\newcommand{\Tpq}{\Indpq{T}}

\newcommandy{\rhythm}[1][1={r}]{{\mathbf{#1}}}
\newcommandy{\rrhythm}[1][1={r}]{\overset{\raisebox{-0.1em}{\tiny$\circ$}}{\rhythm[#1]}}
\newcommandy{\labelling}[1][1={\lambda}]{{\boldsymbol{#1}}}
\newcommandy{\signature}[1][1={s}]{{\boldsymbol{#1}}}
\newcommandy{\slabelling}{\labelling[\sigma]}

\newcommandy[o]{\val}[2][1=b]{%
  \ifthenelse{\equal{#2}{}}%
  {\pi}%
  {%
    \ifthenelse{\equal{#1}{b}}
      {\pi\hspace*{-0.1em}\left(#2\right)}%
      {\pi(#2)}%
  }
}

\newcommand{\intint}[2]{\llbracket#1,#2\rrbracket}

\newcommandy{\pset}[2][1=n,2=R]{E_{#1}^{#2}}

\newcommandy{\rcod}[2][2=r]{\cod{#1}_{\rhythm[#2]}}

\newcommandWithStar{\Tr}{\mathcal{T}_{\rhythm}}
\newcommandWithStar{\taupq}{\Indpq{\tau}}
\newcommandWithStar{\pq}{\frac{p}{q}} 
\newcommandWithStar{\Lpq}{\Indpq{L}}
\newcommandWithStar{\Lpqbar}{\widebar{\Indpq{L}}}
\newcommandWithStar{\Lr}{L_{\rhythm}}
\newcommandWithStar{\Mr}{M_{\rhythm}}
\newcommandWithStar{\Kr}{K_{\rhythm}}
\newcommandWithStar{\Ar}{\mathcal{A}_{\rhythm}}
\newcommandWithStar[o]{\Ap}{A_p}
\newcommandWithStar{\Aps}{A_p^{\xmd *}}
\newcommandWithStar{\Aq}{A_q}
\newcommandWithStar{\Aqs}{A_q^{\xmd *}}
\newcommandWithStar{\negr}{\rhythm[s]_{\frac{p}{q}}}
\newcommandWithStar{\negL}{\widebar{\Lpq}}
\newcommandWithStar{\rpq}{\rhythm_{\frac{p}{q}}}
\newcommandWithStar{\canonlab}{\labelling{\gamma}_\frac{p}{q}}
\newcommandWithStar{\negl}{\labelling[\mu]_{\frac{p}{q}}}
\newcommandWithStar[o]{\N}{\mathbb{N}}
\newcommandWithStar[o]{\Nmu}{\mathbb{N}\setminus\set{0}}
\newcommandWithStar[o]{\Z}{\mathbb{Z}}
\newcommandWithStar[o]{\Tpq}{\Indpq{\mathcal{T}}}
\newcommandWithStar{\Tpqbar}{\widebar{\Tpq}}
\newcommandWithStar{\treer}{\rtree{\rhythm}}
\newcommandWithStar{\Lrl}{L_{\rhythm,\labelling}}

\newcommand{\rtree}[1]{\mathcal{J}_{#1}}

\newcommand{\jpu}{{j~\!\!\text{\psscalebox{0.95}{$+$}\psscalebox{0.90}{$1$}}}}

\newcommandx{\yaHelper}[2][1=\empty]{%
\ifthenelse{\equal{#1}{\empty}}%
  { \ensuremath{\scriptstyle{#2}}} 
  { \raisebox{ #1 }[0pt][0pt]{\ensuremath{\scriptstyle{#2}}}}  
}

\newcommandx{\yrightarrow}[4][1=\empty, 2=\empty, 4=\empty, usedefault=@]{%
  \ifthenelse{\equal{#2}{\empty}}
  { \xrightarrow{ \protect{ \yaHelper[#4]{#3} } } } 
  { \xrightarrow[ \protect{ \yaHelper[#2]{#1} } ]{ \protect{ \yaHelper[#4]{#3} } } } 
}

\newcommandy[o]{\pathx}[2][2={}]{
  \nlb%
  \ifthenelse{\equal{#2}{}}%
    {\yrightarrow{\ #1\ }[-1pt]}%
    {\underset{#2}{\xrightarrow{\ #1 \ }}}%
  \nlb%
}

\newcommand{\itree}{i-tree\xspace}
\newcommand{\itrees}{i-trees\xspace}

\newcommand{\Itree}{I-tree\xspace}

\newcommand{\vmcard}[1]{{\tt Card}(#1)}

\makeatletter

\renewcommand*{\thefigure}{\arabic{figure}}
\let\c@table\c@figure
\makeatother 

\nolinenumbers

\ifdraft
  \linenumbers
  \usepackage[notref,notcite]{showkeys}
  
\fi

\ifdraft\ShowFrame\fi

\begin{document}
\thispagestyle{plain}
\tolerance=5000

\ifdraft
  \setcounter{page}{-1}
  \vspace*{0cm}
  \clearpage
  \vfill
  \setcounter{tocdepth}{2}
  \tableofcontents
  \thispagestyle{plain}
  \vfill
  \vfill
  \clearpage

\fi
\maketitle
\thispagestyle{plain}
\pagestyle{plain}
\begin{abstract}
  We present here the notion of \emph{breadth-first signature} 
    and its relationship with numeration system theory.
  It is the serialisation into an \emph{infinite word} 
    of an ordered infinite tree of finite degree.
  We study which class of languages corresponds to which class of words and,
    more specifically, using a known construction from numeration system theory,
    we prove that the signature of rational languages are substitutive sequences.
\end{abstract}



\section{Introduction}
This work introduces a new notion: 
  the breadth-first signature of a tree (or of a language).
It consists of an infinite word describing the tree (or the language).
Depending on the direction (from tree to word, or conversely), it is either
  a \emph{serialisation} of the tree into an infinite word
  or a \emph{generation} of the tree by the word.
We study here the serialisation of rational, or regular, languages.
 

The (breath-first) signature of an ordered tree of finite degree
  is a sequence of integers, the sequence of the
  degrees of the nodes visited by a breadth-first traversal of the tree.
Since the tree is ordered,
  there is a \emph{canonical} breadth-first traversal; 
  hence the signature is uniquely defined and characteristic of the tree.
  
Similarly, we call \emph{labelling} the infinite sequence
  of the labels of the edges visited by 
  the breadth-first traversal of a labelled tree.
The pair signature/labelling is once again characteristic of the labelled tree.
It provides an effective serialisation of labelled
  trees, hence of prefix-closed languages.


\smallskip

The serialisation of a (prefix-closed) language is very close,
  and in some sense, equivalent to the enumeration of the words
  of the language in the radix order.
It makes then this notion particularly fit to describing the 
  languages of integer representations in various numeration systems.
It is of course the case for the representations in an integer base~$p$ 
  which corresponds to the signature~$p^\omega$,
  the constant sequence.
But it is also the case for non-standard numeration systems such as the 
  Fibonacci numeration system whose representation language has for 
  signature the Fibonacci word (\cf \rsection{subst});
  and the rational base numeration systems as defined in \cite{AkiyEtAl08}
  and whose representation languages have periodic signatures, that is,
  signatures that are infinite periodic words.
To tell the truth, it is the latter case that first motivated our study of 
  signatures.
In another work still in preparation~\cite{MarsSaka2014a}, we study trees
  and languages that have periodic signatures.
  
In the present work, we first study in detail the notion of signature of trees
  (\rsection{ontrees}) and of languages (\rsection{language}).
Then, in \rsection{subst}, we give with \rtheorem{subt<->rat} a 
  characterisation of the signatures of (prefix-closed)
  rational languages as those whose signature is a substitutive sequence.
The proof of this result relies on a correspondence between 
  substitutive sequences and automata due to
  Maes and Rigo \cite{RigoMaes2002} and whose principle 
  goes back indeed to the work of
  Cobham \cite{Cobham1972}.

\section{On trees}\lsection{ontrees}

Classically, trees are undirected graphs in which any two vertices are 
  connected by exactly one path (\cf \cite{Dies97}, for instance).
Our view differs in two respects.

First, a tree is a \emph{directed} graph~$T=(V,\Gamma)$ such that there exist a 
  \emph{unique} vertex, called \emph{root}, which has no incoming arc,
    and there is a \emph{unique (oriented) path} from the root to every 
    other vertex.
Elements of the tree~$T$ gets particular names: 
  vertices are called \emph{nodes}; 
  if~$(x,y)$ is an arc,~$y$ is called \emph{a child} of~$x$ and~$x$ 
  \emph{the father} of~$y$; 
  a node without children is a \emph{leaf}.
We draw trees with the root on the left and arcs rightward.

Second, our trees are \emph{ordered}, that is, that there is a total order 
  on the set of children of every node. 
The order will be implicit in the figure, with the convention that lowermost 
  children are the smallest (according to this order).
In one word, the two trees of \rfigure{orderedtrees} are different (non-isomorphic).

The class of trees we consider is quite close to the one from \cite{tata2007}, 
  but our approach differs greatly. 
They generate trees through tree automata, a depth-first process while we 
  describe them in a breadth-first manner.

\begin{figure}[ht!]
  \centering
  \includegraphics[width=0.9\linewidth]{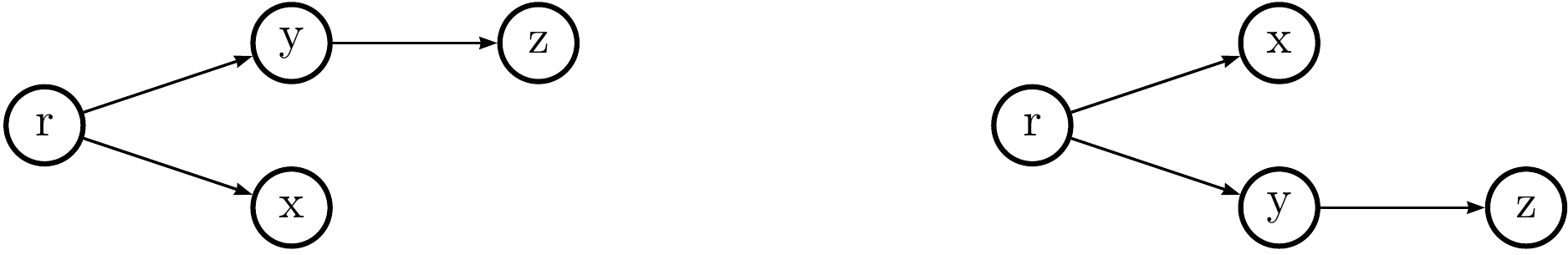}
  \caption{Two non-isomorphic trees}
  \lfigure{orderedtrees}
\end{figure}

The \emph{degree} of a node is the number of its children; it may be finite or infinite.
A tree is of \emph{bounded degree} (resp. \emph{finite degree}) if the 
  degree of every node is bounded (resp. finite).
In the following, we deal with infinite trees of bounded degree only.
Even though most definitions would still work for infinite trees of finite degree,
  this more general setting has no use when considering languages, as we will 
  in most of the article.

\subsection{Relational definition of trees}\lsection{reldeftree}

Given a particular tree, the breadth-first traversal
  naturally and uniquely (since the children of every node are ordered) 
  defines a total ordering of its nodes.
We may then consider that the set of nodes of a tree is always the set of 
  integers~$\N$; 
  the node~$n$ of~$\N$ being the~$(n+1)$-th node visited by the traversal.


\begin{proposition}\lproposition{tree-def}
  A directed graph~$(\N,\theta)$, where the relation~$\theta:\N\rightarrow\N$ satisfies the two conditions
  \begin{enumerate}[label={\normalfont(\roman{*})}]
    \item $\theta$ is injective;
    \item $\forall n\in\N$, $\exists m\in\N, \ee m>n \e $ and $\e  \theta(\intint{0}{n})=\intint{1}{m}$;
  \end{enumerate}
  is an infinite ordered tree of finite degree, written~$T_\theta$.
\end{proposition}
\begin{proof} 
In this setting,~$\theta$ is the child relation,~$0$ is the 
  root,~$\theta(0)=\theta(\intint{0}{0})=\intint{1}{k}$, is an interval of \N*; 
  it is the (finite) ordered set of the~$k>0$ children of the root. 
Given a \emph{positive} integer~$n$, 
  $$ \theta(n) = \theta(\intint{0}{n})\setminus\theta(\intint{0}{n-1}) $$
is the (possibly empty) interval of \N*  of the children of the node~$n$.

\noindent Hence the \emph{father relation}~$\theta^{-1}$ satisfies the following properties:
  \begin{enumerate}
    \item $\theta^{-1}$ is a function --- from (i);
    \item Dom$(\theta^{-1})= \N_+$ --- from (ii);
    \item $\theta^{-1}(n) < n$ --- from (ii).
  \end{enumerate}
  It then yields a unique path (in~$\theta^*$) from the root to every 
    vertex in~$\N_+$ 
\end{proof}

\paragraph{Computing the relation from the tree.} 
A breath-first traversal of an infinite ordered tree~$T$ of finite 
  degree inductively maps the set of nodes of T onto \N* and builds a child 
  relation~$\theta$ 
  by the following procedure whose principle is essential.

The root of~$T$ is mapped onto~$0$, the ordered set of the ~$k$ children of the 
  root is mapped onto the interval~$\intint{1}{k}$, that 
  is~$\theta(0)=\intint{1}{k}$ and two integer indices are set: 
  the first one represents \emph{the node to be treated}, call it~$n$, and is set to~$1$;
  the second one represents \emph{the last node created}, call it~$m$, and is set to~$k$.
At every step of the procedure the node~$n$ is considered the ordered set of its~$k_n$ children is mapped
  onto the interval~$\intint{m+1}{\,m+k_n}$, that is~$\theta(n)=\intint{m+1}{\,m+k_n}$ 
  (possibly empty if~$k_n=0$);
  then~$n$ is incremented by~$1$,~$m$ by~$k_n$, and the procedure takes on a new step.
  
Since~$T$ is of finite degree, each step is well-defined and since~$T$ is infinite,
  the procedure never ends.
Nevertheless,~$\theta(n)$ is eventually defined for every integer~$n$.
The way it is defined makes $\theta$ meet Conditions $(i)$ and~$(ii)$
  of \rproposition{tree-def} and the tree~$T_\theta$ is isomorphic to~$T$.
  

\paragraph{On \itrees.}
It will prove to be extremely convenient to have a slightly different look at trees 
  and to consider that the root of a tree is also a \emph{child of itself} 
  that is, bears a loop onto itself.
It amounts to changing the Condition~(ii) of the child relation~$\theta$ to
\smallskip
\begin{itemize}
  \item[(ii')] $\forall n\in\N$, $\exists m\in\N, \ee m>n \e $ and $\e  \theta(\intint{0}{n})=\intint{0}{m}$\,;
\end{itemize}
\smallskip
the difference being that the interval~$\intint{1}{m}$ of~(ii) is changed to~$\intint{0}{m}$ in~(ii'). 

It should be noted that this convention is sometimes taken when implementing
  tree-like structures (for instance the unix/linux file system).
It implies that the father relation~$\theta^{-1}$ is now a 
  \emph{function}~$\N\rightarrow\N$ and will 
  make the connexion with numeration systems very natural.

Of course, a graph~$T_\theta$ defined by a relation~$\theta$ that meets 
  Condition~$(i)$ and~$(ii')$ is not formally a tree;
  we call such structures \emph{\itrees}.
It is so close to a tree that we pass from tree to \itree (or conversely)
  with no further ado.
%
%
%
%

\subsection{Breadth-first signature of a tree}

\begin{definition}[Breadth-first signature of a tree]\ldefinition{signature}
  Given a tree of child relation~$\theta$, we call \emph{breadth-first signature} or, for short, \emph{signature} of~$\theta$ 
  the infinite integer sequence
  \begin{subequations}\lequation{signature}
    \begin{align}
      \signature = s_0 s_1 \cdots s_k \cdots \ee \text{where~} s_i &= \vmcard{\theta(i)} \ee \text{ for all integers }i>0 \\  
     \text{and~} s_0 &= \vmcard{\theta(0)}+1 \lequation{n0}
    \end{align}
  \end{subequations}
\end{definition}
where~$\vmcard{X}$ is the cardinal of the set~$X$.
It follows directly from this definition that the breadth-first signature 
  is characteristic of its tree, as stated below.
\begin{proposition}
  Two trees with the same breadth-first signature are equal.
\end{proposition}

The special case of~0 (\cf~\requation{n0}) is an artefact
  of the non-surjectivity of~$\theta$ already discussed in the previous 
  \rsection{reldeftree}.
For short, the signature is more canonically associated with an \itree than with the corresponding tree.

\subsection{Generating a tree by its signature}\lsection{sig->tree}
A signature~$\rhythm[s]=s_0 \xmd s_1 \cdots s_k \cdots$
  is \emph{valid} if it satisfies the following equation
\begin{equation}\label{eq.validsignature}
  \forall j \in \N \quantsp \sum_{i=0}^{j} s_i ~>~ \jpu\eqpnt
\end{equation}

This restriction ensures that the sequence is indeed the signature of a
  tree, as stated below; if it were not the case, one could still apply the
  procedure hereafter, but the resulting graph would not be connected.~\footnote{\requation{validsignature} is the counterpart of 
  the `$m>n$' condition in \rproposition{tree-def}(ii).}
  
\begin{proposition}
  For every valid signature~$\signature$, 
    there exists a unique tree whose signature is equal to~$\signature$.
\end{proposition}

We will describe the tree whose signature is equal to a given 
  signature~$\signature$ by enumerating its edges in the breadth-first order.
It is essentially the reverse as the construction of the relation~$\theta$
  from the respective tree, given at~\rsection{reldeftree}

We maintain two integers: the starting point~$n$ and the end point~$m$ of the transition.
In one step of the algorithm,~$s_n$ nodes are created, corresponding to 
  the integers~$m,m+1,\ldots,{(m+s_n-1)}$, and~$s_n$ edges are created
  (all from~$n$, and one to each of this new nodes).
  Then~$n$ is incremented by~1, and~$m$ by~$s_n$.

The validity of~$\signature$ ensures that, at all point~$n<m$, hence that every 
  node has a father smaller than itself.
\rfigure{T321}, in appendix, shows the first few steps of the procedure for the purely 
  periodic signature~$(321)^\omega$, while \rfigure{t321_f} shows the 
  resulting \itree.

\begin{figure}[ht!]
\centering
\subfloat[The \itree~$T_{\signature}$]{\lfigure{t321_f}\includegraphics[width=0.47\linewidth]{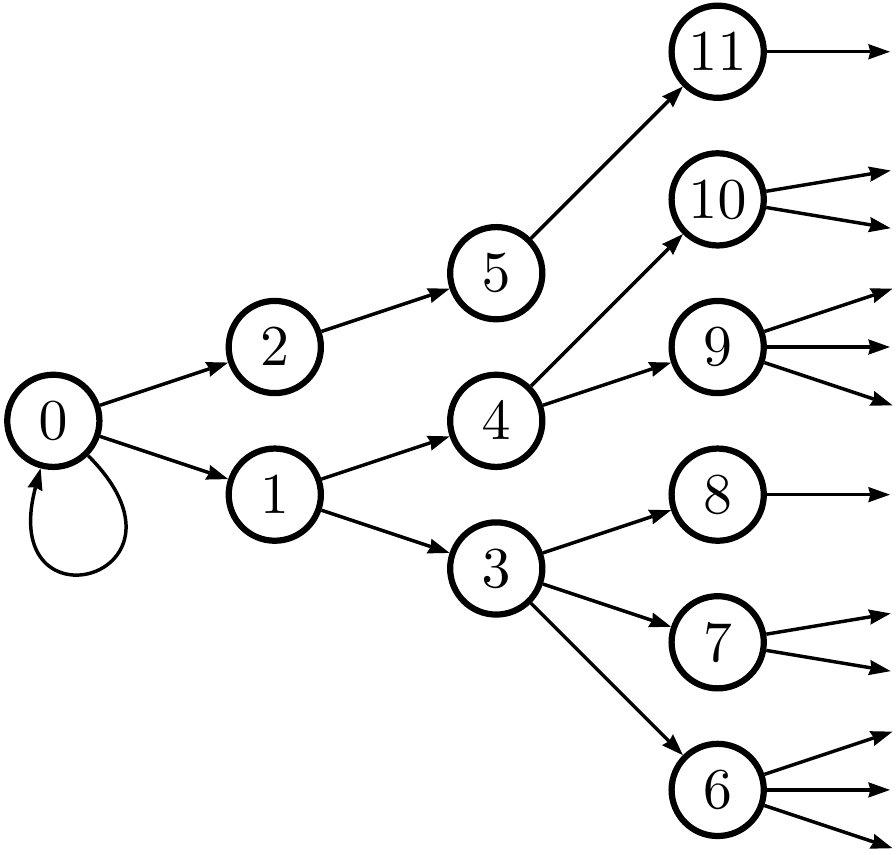}}
\hfill
\subfloat[$L_{(\signature,\labelling) },with~\labelling = abd\xmd .\xmd bc \xmd . \xmd a \xmd . \xmd bcd \xmd . \xmd ad \xmd . \xmd d \xmd \ldots$]{\lfigure{l321}\hspace*{0.5cm}\includegraphics[width=0.47\linewidth]{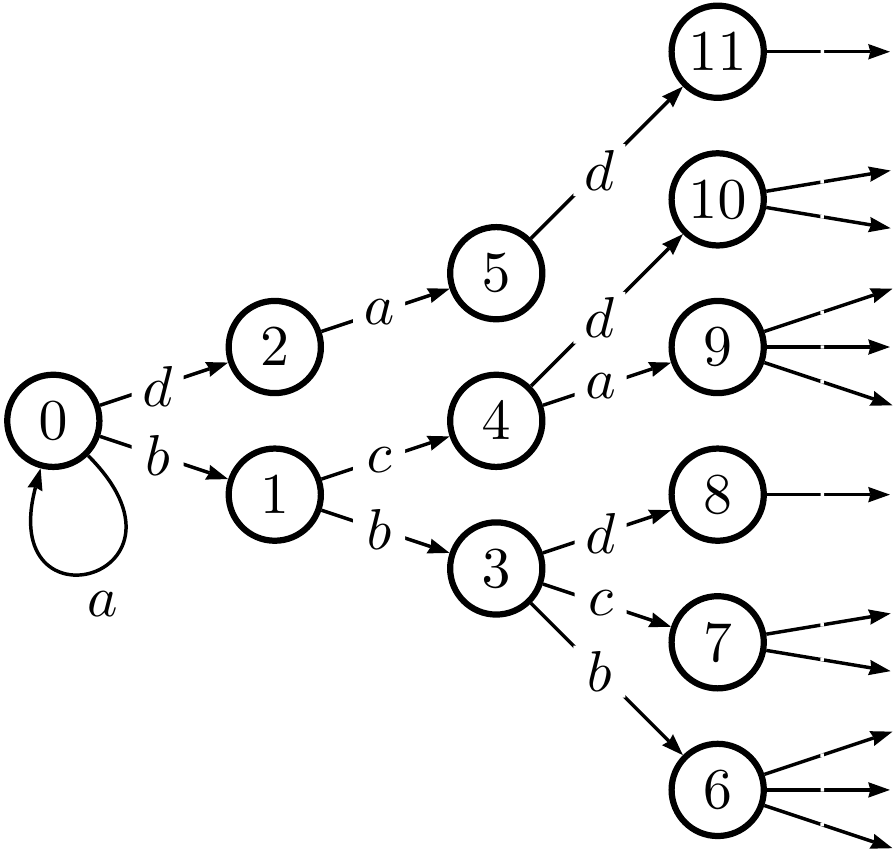}}
\caption{\Itree generated by the signature~$\signature=(321)^\omega$}
\lfigure{T321_f}
\end{figure}

\newcommand{\dynvmscale}{0.8}
\begin{figure}[p!]
  \centering
  \subfloat[$s_0=\mathbf{3}$]{\includegraphics[scale=\dynvmscale]{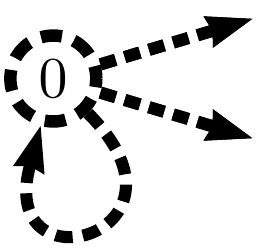}}
  \hfill
  \subfloat[$s_1=\mathbf{2}$]{\includegraphics[scale=\dynvmscale]{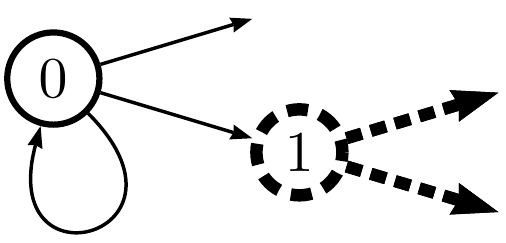}}
  \hfill
  \subfloat[$s_2=\mathbf1$,]{\includegraphics[scale=\dynvmscale]{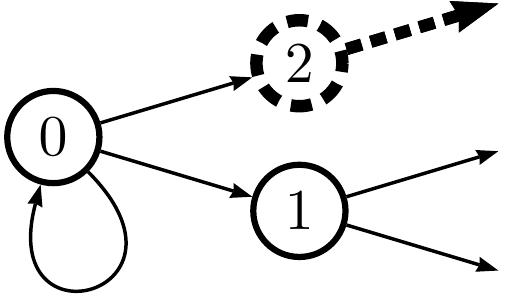}}
  \hfill
  \subfloat[$s_3=\mathbf3$]{\includegraphics[scale=\dynvmscale]{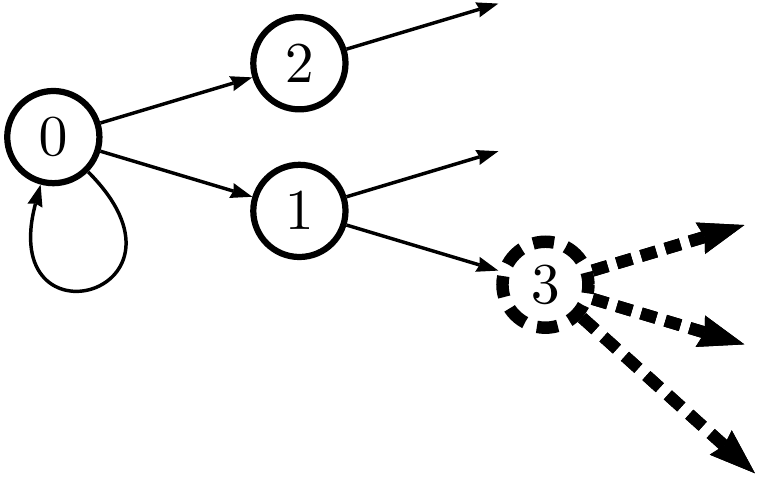}}
  \hfill
  \subfloat[$s_4=\mathbf2$]{\includegraphics[scale=\dynvmscale]{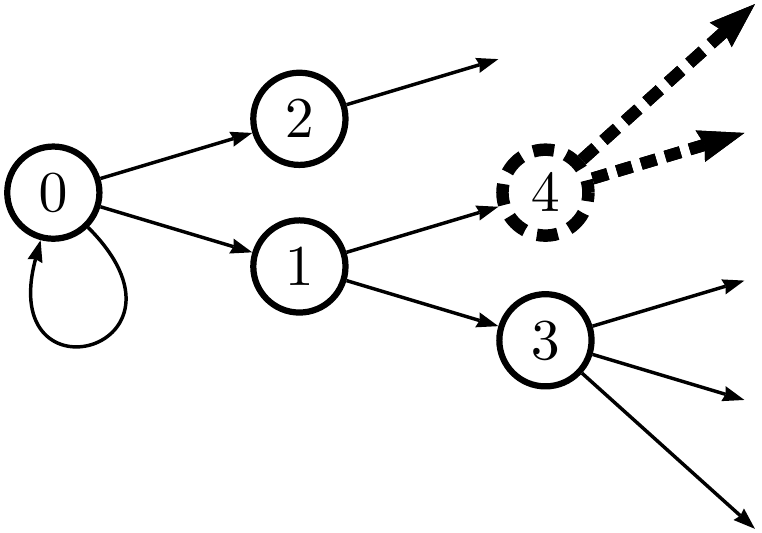}}
  \hfill
  \subfloat[$s_5=\mathbf1$]{\includegraphics[scale=\dynvmscale]{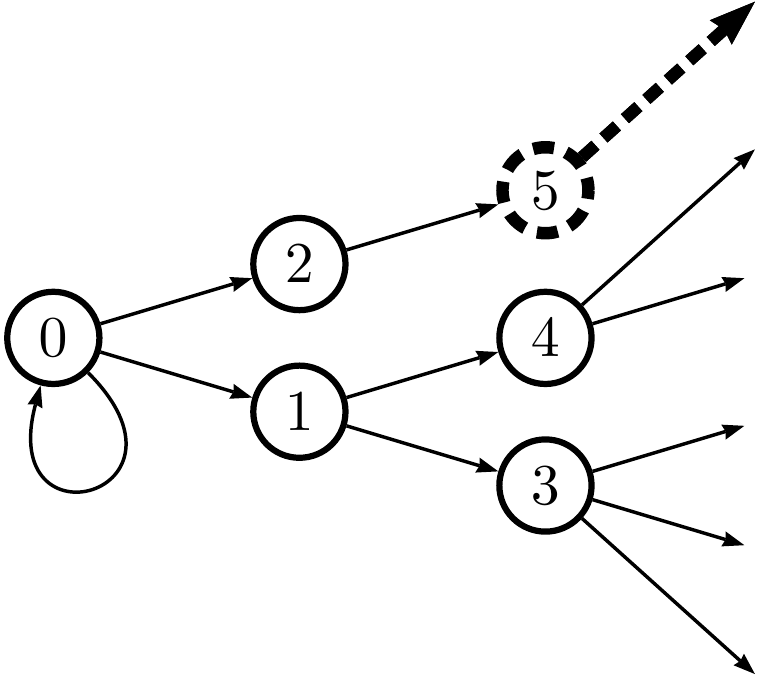}}
  \hfill
  \subfloat[$s_6=\mathbf3$]{\includegraphics[scale=\dynvmscale]{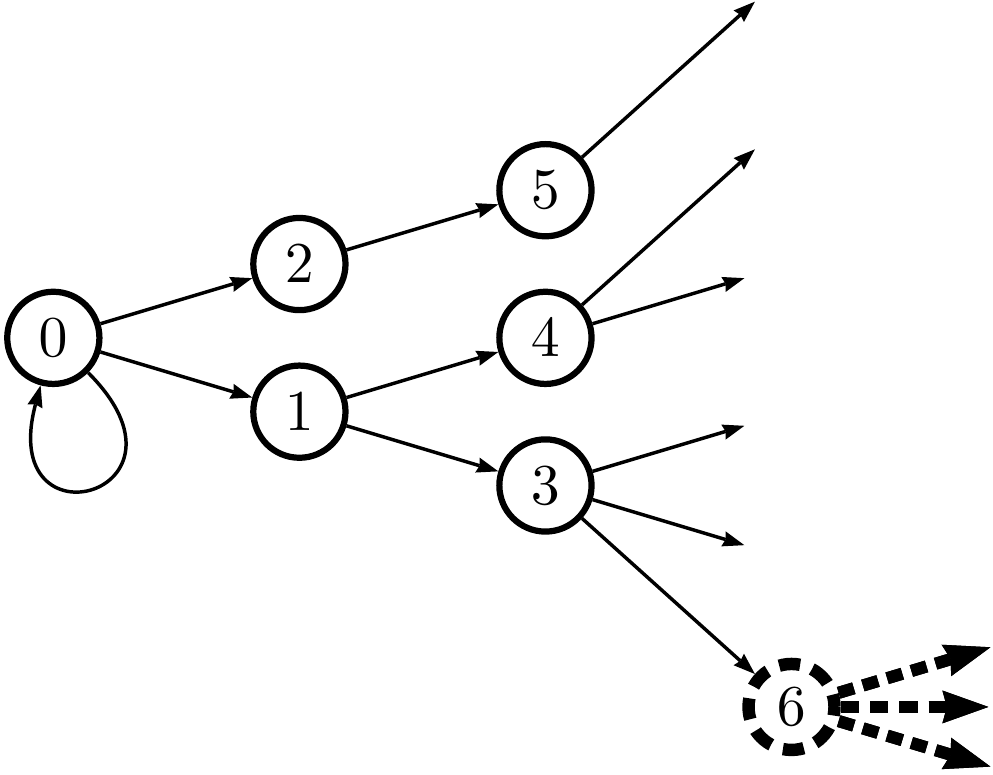}}
  \hfill
  \subfloat[$s_7=\mathbf2$]{\includegraphics[scale=\dynvmscale]{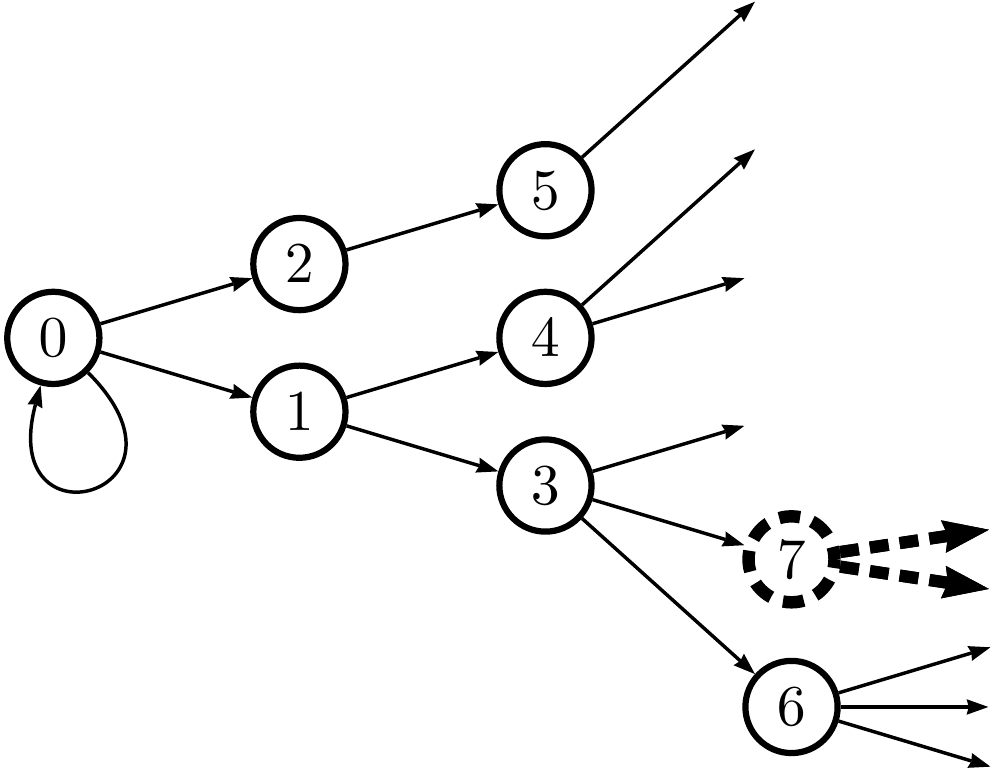}}
  %
  \caption{The first eight steps of the generation of~$T_{(321)^\omega}$}
  \lfigure{T321}
\end{figure}


\section{Signature for languages}\lsection{language}
  An \emph{alphabet} is a set of \emph{letters} and will always be ordered in the following.
  Whenever we use a latin or digits alphabet, it will be ordered as usual
    (that is,~$a<b<c<\cdots$ or~$0<1<2<\cdots$).
  A \emph{word}~$w$ is a finite sequence of letters~$a_0\xmd a_1\cdots a_{n-1}$ and
    its length is denoted by~$\wlen{w} = n$.
\subsection{Labelling}
A labelling, together with a signature~$\signature$, is the description of a labelled
  tree (that is, essentially a prefix-closed language). 
It corresponds to the sequence of the transition labels of the tree, taken in
  breadth-first order.
It follows that a \emph{labelling} is simply a sequence of letters of some alphabet.
%
%

However, for a labelled tree to effectively represent a (prefix-closed) language, it must satisfies 
  some properties.
For instance, two edges with the same starting point must have distinct labels.
More generally, the labels must be consistent with the  
  breadth-first traversal: an edge to a \emph{smaller} child
  must be labelled by a \emph{smaller} letter.
The notion validity for a labelling subdues these issues.
%
%


Given a signature~$\signature$, a labelling~$\labelling$ over an alphabet~$A$ is \emph{valid} 
  (with respect to~$\signature$) 
  if there exists a family~$\left\lbrace w_k \right\rbrace_{k\in\N}$ of words over~$A$ such that 
  \begin{enumerate}
    \item $\labelling$ is the concatenation of the family~$\left\lbrace w_k \right\rbrace_{k\in\N}$: 
      $$\labelling = w_0\xmd w_1 \xmd \cdots \xmd w_k\cdots\eqpntvrg$$ 
    \item the length of each word~$w_k$ is equal to~$s_k$:
      $$ \forall k\in\N \quantsp \wlen{w_k} = s_k \eqpntvrg $$
    \item the letters of each word~$w_k$ are in strictly increasing order:
    $$ \forall w_k = a_0 a_1 \cdots a_n \quantsp a_0 < a_1 < a_2 < \cdots < a_n \eqpnt $$
  \end{enumerate}
For instance if the signature starts with~$3 \xmd 2 \xmd 1 \xmd 3 \cdots$, a valid labelling could
  start with~$abd\xmd .\xmd bc \xmd .\xmd a \xmd . \xmd bcd\ldots$\xmd; or 
  with~$012\xmd .\xmd01\xmd .\xmd0\xmd .\xmd012\ldots$
%
%
A pair signature/labelling~$(\signature,\labelling)$ 
  is called a \emph{labelled signature}; it is \emph{valid}
  if both~$\signature$ is valid and~$\labelling$ is valid (w.r. to~$\signature$).

A valid labelled signature~$(\signature,\labelling)$ 
  uniquely defines a labelled tree, by using a 
  procedure analogous to the one from
  \rsection{sig->tree}. Every edge~$i\pathx{}j$ created 
  is labelled by~$\lambda_j$.
For every node~$n$, we denote by~$\cod{n}_{(\signature,\labelling)}$ 
  the word labelling the unique
  path~$0\pathx{~} n$.
We denote by~$L_{(\signature,\labelling)}$ the language of such words: \footnote{%
  This process closely related to the creation of an abstract numeration systems (\cf \cite{LecoRigo10hb})
  which takes a language~$L$ (usually assumed to be rational) and set the representation of~$n$ in the new numeration system
  as the~$(n+1)$-th word of~$L$ in radix order.
}
  $$  L_{(\signature,\labelling)} = \set{ \cod{n}_{(\signature,\labelling)}~|~n\in\N } \eqpnt $$
\rfigure{l321}, \pfigure{l321}, shows the language whose signature is~$\signature=(321)^\omega$ and labelling starts 
  with~${\labelling = abd\xmd .\xmd bc \xmd . \xmd a \xmd . \xmd bcd \xmd . \xmd ad \xmd . \xmd d \xmd \ldots}$
The validity of the labelled signature insures that the words~$\cod{n}_{(\signature,\labelling)}$ are all distinct, 
  hence the following lemma.
\begin{lemma}
  The~$(n+1)$-th word of~$L_{(\signature,\labelling)}$ in radix order is~$\cod{n}_{(\signature,\labelling)}$
\end{lemma}

%




%
Conversely, given any prefix-closed language~$L$ over an alphabet~$A$, 
  there is a unique valid labelled signature~$(\signature,\labelling)$ 
  generating it;~$\signature$ is defined by the underlying tree of~$L$
  and~$\labelling$ is the sequence of the labels of the 
  edges of the underlying tree of~$L$
  when taken in breadth-first order.
The next statement follows immediately.
  
  

\begin{proposition}
  For every valid labelled signature~$(\signature,\labelling)$, 
    there exists a unique language whose signature is equal to~$(\signature,\labelling)$.
\end{proposition}

\subsection{Minimal labelling and rational trees}
We call \emph{minimal labelling} of a signature~$\signature$ (or equivalently of a tree~$T_{\signature}$)
  the labelling induced by the order of children:
  $$ \labelling[\mu] = w_0\xmd w_1 \xmd \cdots \xmd w_k\cdots \ee \text{where} 
  \e \forall k\in\N \e w_k = 0 \xmd 1\xmd 2 \cdots n \e \text{and} \e n = (s_k-1) \eqpnt$$
Intuitively, it corresponds to add labels in the tree such that 
  the transition~$n\pathx{a}m$ is labelled by~$a=0$ if~$m$ is the 
  \emph{smallest} child of~$n$, and that the transition~$n\pathx{b}(m+k)$ 
  is labelled by~$b=k$, if it exists.
It is always possible to label a tree in such a way and 
  it produces a valid labelled signature.
%
%
Intuitively, the minimal labelling is the simplest way to label a tree, in the sense that 
  it adds the less possible complexity.
The next lemma gives an example of this intuition.
\begin{lemma}
  Let $(\signature,\labelling)$ be a valid labelled signature, 
    and~$\signature[\mu]$ the minimal labelling associated with~$\signature$.
  If~$L_{(\signature,\labelling)}$ is a rational language, then~$L_{(\signature,\labelling[\mu])}$ is rational as well.
\end{lemma}
\begin{proof}
  Given the finite deterministic trim 
    automaton~$\Ac=\aut{A,Q,\delta,i}$ accepting~$L_{(\signature,\labelling)}$,
    let us consider the 
    automaton~$\Bc=\aut{B,Q,\delta',i}$ where
  \begin{itemize}
    \item $B = \intint{0}{k}$  ~~~ with~$k = \vmcard{A}$
    \item $p \pathx{i}[\Bc] q$~~~iff~~~$p \pathx{b}[\Ac] q$ and there exists exactly~$i$ letters $a$ of $A$ such that
    \begin{itemize}
      \item $a<b$ 
      \item $p\pathx{a}[\Ac]q'$ for some state~$q'$
    \end{itemize}
  \end{itemize}
  Intuitively, one has to change the labels of the outgoings transitions of every 
    states by the smallest possible (in~$\intint{0}{k}$) without modifying their relative order.
  For instance if a state~$p$ of~$\Ac$ has three outgoings transitions labelled by $a$,~$c$ and~$d$;
    then in the automaton~$\Bc$, the same state $p$ would have the same transitions but now respectively labelled by~$0$,~$1$ and~$2$ (provided that the order of~$A$ is~${a<c<d}$). 
  See \rfigure{minlab}, below, for an example.  
  Unfolding automata~$\Ac$ and~$\Bc$ into infinite labelled trees yields the statement.
  \begin{figure}[ht!]
    \hfill
    \subfloat[An automaton~$\Ac$]{\includegraphics[width=0.35\linewidth]{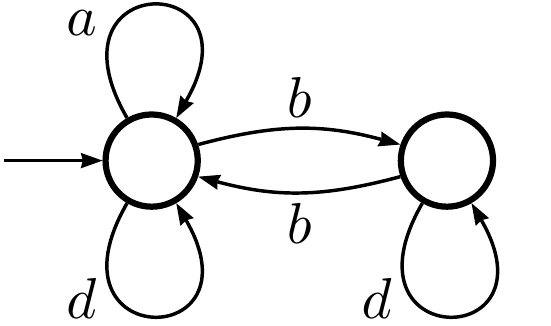}}\hfill
    \subfloat[The respective automaton~$\Bc$]{\hspace*{2cm}\includegraphics[width=0.35\linewidth]{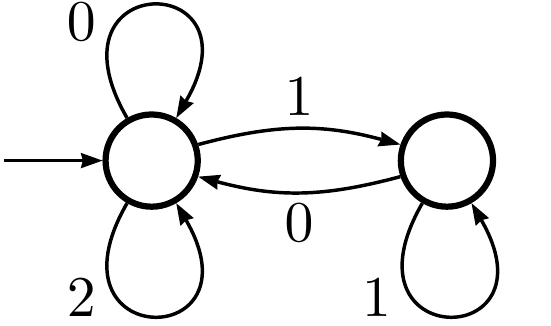}}
    \hfill\hspace*{0cm}
    \caption{Minimal labelling}
    \lfigure{minlab}
  \end{figure}
\end{proof}

\begin{remark}\lremark{lab->complexity}
  It should be noted that even if a signature produces a really simple tree
  (such as the infinite unary tree), one can always choose a labelling in order to
  produce an artificially complex language 
  (such as the infinite word where the~$i$-th letter is a~$1$ if 
  the~$i$-th Turing machine stops on the empty word).
  
  This is why positive results relative to the regularity of a language defined
    by signature will always require some restriction on the labelling.
  It usually amounts to ensure that signature and labelling are generated 
    in similar fashions.
  For instance, it will be the case for 
    \emph{substitutive labelled signature} defined in the next
    \rsection{subst}.
\end{remark}

\section{Substitutive signature and rational languages}\lsection{subst}
  
  The purpose of this section is to establish a relationship between
    substitutive sequences and rational languages.
  Let us first consider the Fibonacci 
    word~$\sigma^\omega(0)$ where~$\sigma(0)=01$ and~$\sigma(1)=0$:
  $$ \sigma^\omega(0) \e = \e 0\xmd1\xmd0\xmd0\xmd1\xmd0\xmd1\xmd0\xmd0\xmd1\xmd0\xmd0\xmd1 \cdots$$
  
  This word, however, is not valid when considered as a signature. 
  We build a valid signature~$\signature$ by replacing the 0's in the Fibonacci word by 2's:
  
  $$ \signature \e = \e 2\xmd1\xmd2\xmd2\xmd1\xmd2\xmd1\xmd2\xmd2\xmd1\xmd2\xmd2\xmd1 \cdots$$
  
  It should be noted that the labelling~$\labelling=\sigma^\omega(0)$ is valid w.r. 
    to~$\signature$: each letter~`2' (resp~`1') of~$\signature$ is associated to the word~`$01$' (resp~`$0$') 
    of~$\labelling$.
  
  The language~$L_{(\signature,\labelling)}$ shown at \rfigure{fibo} is then exactly
    the integer representations in the
    Fibonacci numeration system (sometimes called 
    Zeckendorf numeration system) 
    that is, the rational language~$1\set{0,1}^* \setminus (\set{0,1}^* 11 \set{0,1}^*)$.
  \begin{figure}
    \includegraphics[width=\linewidth]{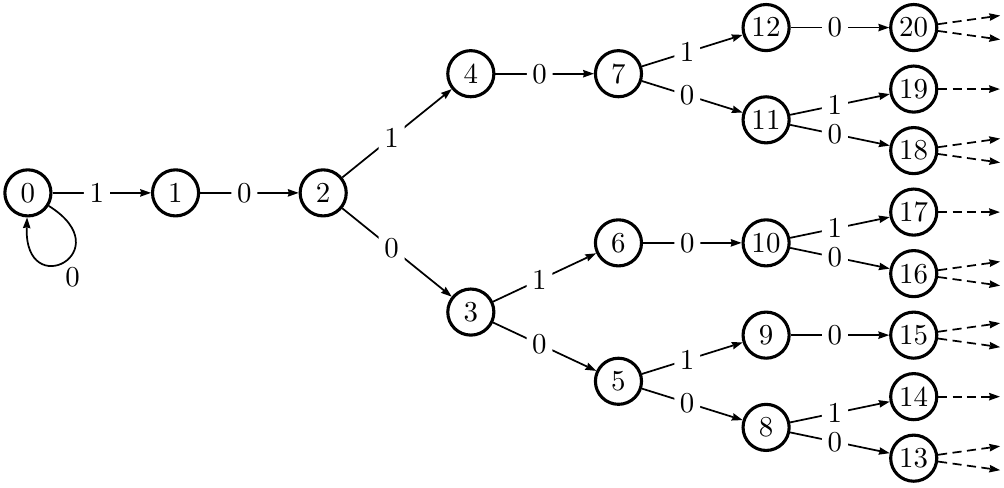}
    \caption{The language of the integer representations in the Fibonacci number system}
    \lfigure{fibo}
  \end{figure}

\subsection{Substitutive sequences and substitutive signatures}
  We recall here some basic definition from combinatorics on words; 
  we essentially use the terminology of~\cite{BertRigo10Ed}.
  
  Given an alphabet~$A$, 
    we say that an endomorphism~$\sigma:A^* \rightarrow A^*$ 
    is \emph{prolongable} on~$a\in A$ if there exists a word~$u$ of~$A^*$
    such that~$ \sigma(a)=au$ and that~$lim_{n\rightarrow+\infty} \wlen{\sigma^n(a)} = +\infty$.
%
  In this context, the sequence~$\sigma^n(a)$ converges (for the usual topology) 
    to an infinite sequence denoted by~$\sigma^\omega(a)$.
  Any sequence resulting from the iteration of a prolongable endomorphism 
    (that is, of the form~$\sigma^\omega(a)$ where~$\sigma$ is prolongable on~$a$) 
    is said to be \emph{purely substitutive}.
    
  A letter-to-letter morphism is called a \emph{coding}%
\footnote{Note that a \emph{coding does not define a code}, in the sense of \cite{BersEtAll2009}.}. 
  The image~$f(w)$ of a purely substitutive sequence~$w$ by a 
    morphism~$f$ is called an 
    \emph{HD0L} sequence;
    if furthermore,~$f$ is a coding,~$f(w)$ is called a \emph{substitutive sequence}.

\medskip

  
  We will now define particular 
    substitutive sequences and consider them as signatures.
  Given an endomorphism~$\sigma:A^*\rightarrow A^*$ prolongable on a 
    letter~${a\in A}$, we denote by~$f_\sigma$ the following 
    coding entirely defined by~$\sigma$:
  $$ \forall a\in A \quantsp f_\sigma(a) = \wlen{\sigma(a)} \eqpnt$$
  We call the substitutive sequence~$f_\sigma(\sigma^\omega(a))$ a \emph{substitutive signature}.

  \begin{lemma}\llemma{substsign->valid}
    Every substitutive signature is valid.
  \end{lemma}
  \begin{proof} 
  It amounts to prove that for all prefixes~$u$ 
    of~$f_\sigma(\sigma^\omega(a))$, the sum of the letters of~$u$ 
    is strictly greater than
    the length of~$u$.
  Hence, from the definition of~$f_\sigma$, that for all prefixes~$v$ 
    of~$\sigma^\omega(a)$,~$\wlen{\sigma(v)}>\wlen{v}$.
  
  Let $v$ be any prefix of~$\sigma^\omega(a)$.
  Since~$\sigma$ is prolongable on the letter~$a$, 
    there is an integer~$i$ such 
    that~$\sigma^i(a) \sqsubseteq v \sqsubset \sigma^{(i+1)}(a)$;
    hence $\sigma^{(i+1)}(a) \sqsubseteq \sigma(v) \sqsubset \sigma^{(i+2)}(a)$,
    hence~$\wlen{\sigma(v)} \geq \wlen{\sigma^{(i+1)}(a)} > \wlen{v}$.

  \end{proof}
  \begin{definition}
  A labelled signature~$(\signature,\labelling)$ is \emph{substitutive} if 
    \begin{itemize}
        \item $\signature$ is a substitutive signature~$f_\sigma(\sigma^\omega(a))$ and
        \item $\labelling$ is of the form~$g(\sigma^\omega(a))$ where~$g:A^* \rightarrow B^*$ and for all letters~$a \in A$, $\wlen{g(a)}=\wlen{\sigma(a)}$.~\footnote{A substitutive labelling is then a particular HD0L sequence.}
    \end{itemize}
  \end{definition}
  
  \noindent The next lemma follows; its proof is essentially the same as the one of~\rlemma{substsign->valid}.
  \begin{lemma}
    Every substitutive labelled signature is valid.
  \end{lemma}
%

  We open now a parenthesis about ultimately periodic signature. 
  Let~$\signature = u v^{\omega}$ be an ultimately periodic sequence over
    the alphabet~$\intint{0}{k-1}$; we
    call \emph{growth ratio} of~$v$, denoted~$gr(v)$, the average of the letters of~$v$:
    $$ gr(v)\e=\e \frac{\sum_{i=0}^{\wlen{v}-1} v[i]}{\wlen{v}} \eqpnt $$
  The next proposition states that whenever~$gr(v)$ is an integer that is, 
    when the sum of the letters of~$v$ is a multiple of the length of~$v$,
    any signature of the form~$uv^\omega$ is substitutive.
  \begin{proposition}\lproposition{gr-int->subst}
    Given an ultimately periodic (valid) signature~$\signature = u v^{\omega}$  
      whose growth ratio is an integer then~$\signature$ is a substitutive signature.
  \end{proposition}
  \begin{proof} 
  First, let us consider the case where~$u=\epsilon$.  
  We denote by~$k$ the length of~$v$:~$k= \wlen{v}$;  
    and consider the alphabet~${A=\intint{0}{k-1}}$.
  In the following, any letter (for instance of the form~$(j+h)$ for some integers~$j$ and~$h$) will be taken
    in~$\Z/k\Z$, hence will belong to~$A$.
  We define the endomorphism~$\sigma:\Ae\rightarrow\Ae$ by
  $$\begin{array}{ll}
  \sigma(0) = 0 \xmd 1 \xmd \cdots (v[0]-1)  &\text{and} \\
  \sigma(i) = (j+1)\xmd (j+2) \cdots \xmd (j+v[i]) \ee &\text{where } j\text{ is the last letter of }\sigma(i-1).
  \end{array}$$
  
  
  Let us now prove that~$\sigma(0\xmd 1 \xmd \cdots \xmd (k-1)) = (0\xmd 1 \xmd \cdots \xmd (k-1))^j$,
    where~$j$ is the growth ratio of~$v$.
  It is quite easy to see that~$\sigma(0\xmd 1 \xmd \cdots \xmd (k-1))$ is a prefix 
    of~$(0\xmd 1 \xmd \cdots \xmd (k-1))^\omega$ since~$\sigma(i+1)$ starts with the 
    letter directly following the last letter of~$\sigma(i)$.
  Since by definition, the length of~$\sigma(i)$ is equal to~$v[i]$ then $$\wlen{\sigma(0\xmd 1 \xmd \cdots \xmd (k-1))}=\sum_{i=0}^{\wlen{v}-1} v[i] = j \times \wlen{v} = j \times k$$ yielding the claim.
  
  It follows that $\sigma^\omega(0\xmd 1 \xmd \cdots \xmd (k-1))$ is equal 
    to~$(0\xmd 1 \xmd \cdots \xmd (k-1))^\omega$.
  With a similar reasoning one can prove 
    that~$\sigma^\omega(0)$ is also equal to~$(0\xmd 1 \xmd \cdots \xmd (k-1))^\omega$.
%
%
  Finally, since for all~$i\in\intint{0}{k-1}$,~$\wlen{\sigma(i)} = v[i]$,
    it follows that~$f_\sigma(0\xmd 1 \xmd \cdots \xmd (k-1)) = v$ and~$f_{\sigma}(\sigma^\omega(0))= f_\sigma((0\xmd 1 \xmd \cdots \xmd (k-1))^\omega) = v^\omega$, concluding the special case~$u=\epsilon$.
  \smallskip
  
  We no longer assume that~$u=\epsilon$ and then denote by~$n$ the length of~$u$.
  The new alphabet of the morphism is~$C = A\uplus B$, where~$B$ is of 
    cardinal~$n$, each letter corresponding to a position in~$u$.
  We denote the letters by : 
    $$ B=\set{b_0, b_1, \ldots, b_{(n-1)}} \ee \text{and} \ee A=\set{a_0,a_1,\ldots,a_{(k-1)}}$$
  The images of the letters of~$B$ by~$\sigma$ are defined inductively:
  for every integer~$i$,~$\sigma(b_0\xmd b_1 \cdots b_i)$ is the prefix of the sequence~$(b_0\xmd b_1 \cdots b_{n-1})(a_0\xmd a_1 \cdots a_{(k-1)})^\omega$ such that~$\wlen{\sigma(b_i)} = u[i]$.
  
  Since the signature~$uv^\omega$ is valid by hypothesis, the last letter of~$\sigma(b_{n-1})$ is some letter of~$A$;
  the image by~$\sigma$ of the letters of~$A$ are then:
    $$\begin{array}{ll}
  \sigma(a_0) = a_{(j+1)} \xmd a_{(j+2)} \xmd \cdots a_{(j+v[0])}  &\text{where } a_j\text{ is the last letter of }\sigma(b_{(n-1)}). \\
  \sigma(a_i) = a_{(j+1)} \xmd a_{(j+2)} \xmd \cdots a_{(j+v[i])} \ee &\text{where } a_j\text{ is the last letter of }\sigma(a_{(i-1)}).
  \end{array}$$
  From here, the proof is the analogous to the case where~$u=\epsilon$.
%
%
%
  
  \end{proof}
 
  \begin{remark}
    It should be noted that whenever the growth ratio of an ultimately periodic signature
      is not an integer, it is never substitutive.
    The proof of this statement is however convoluted and is the subject of 
      another article in preparation \cite{MarsSaka2014a}.
      
    It should also be noted that ultimately periodic signatures 
      are, as words, \emph{purely substitutive} no matter the growth ratio.
    For instance the word~$(21)^\omega$ is equal 
      to~$\sigma^\omega(2)$ where~$\sigma(2)=\sigma(1)=21$.
    It illustrates the fact that the set of purely substitutive 
      sequences is not included in the set of substitutive signatures.
      
  \end{remark}


  \subsection{Rational languages and substitutive signatures}
  
  \begin{theorem}\ltheorem{subt<->rat}
    A prefix-closed language is rational if and only if its labelled signature is a substitutive signature.
  \end{theorem}
  
  The proof of this theorem relies on a transformation from finite automaton 
    to substitutive word used by Rigo and Maes in \cite{RigoMaes2002} 
    (\cf also \cite[Section~3.4]{LecoRigo10hb}) to prove the equivalence between 
    two decision problems: 
    1- the ultimate periodicity in an abstract numeration system 
      (\cf \cite{LecoRigo01} or \cite{LecoRigo10hb}) 
    and 
    2- the ultimate periodicity problem of an HD0L word 
      (solved independently in \cite{Mitrofanov2011} and \cite{Durand2013a}).
    
  \begin{table}[ht!]
  \centering
  \begin{tabular}{|c|c|}
    \hline
         Automaton & Substitutive signature \\
         $\aut{\Sigma,Q,\delta,i}$ & $\signature=f_\sigma(\sigma^\omega(a))$~~$\labelling=g(\sigma^\omega(a))$\\
    \hline
           & $\sigma:A^*\rightarrow A^*$ \\
           & $g:~A\rightarrow B$ \\
       $Q$ & $A$ \\
       $i$ & $a$ \\
       $\Sigma$  & $B$ \\
      \multirow{2}{*}{$(b,x,c)\in \delta$} & the $k$-th letter of $\sigma(b)$ is~$c$ \\
                          & the $k$-th letter of $g(b)$ is~$x$\\
    \hline
  \end{tabular}
  \vspace*{0.1cm}
  \caption{Summary of the transformation~~DFA~$\leftrightarrow$ Substitutive signature}
  \ltable{sumup}
  \end{table}
  
  \begin{proposition}\lproposition{subst->rat}
    Given a valid substitutive signature~$(\signature,\labelling)$, the 
      language~$L_{(\signature,\labelling)}$ is a rational language.
  \end{proposition}
  \begin{proof}
    We denote by~$\sigma$ the endomorphism~$A^*\rightarrow A^*$ prolongable on a letter~$a$ of~$A$;
      and by~$g$ the projection~$A^*~\rightarrow~B^*$ such that
           $$ \signature=f_\sigma(\sigma^\omega(a)) \ee \text{and} \ee\labelling=g(\sigma^\omega(a)) \eqpnt$$
    Since we are using two alphabets at the same time we will, in this proof, consider 
      that~$a,b,c$ are letters of~$A$ and $x$ is a letter of~$B$.
      
    We consider the automaton~$\Ac=\aut{A,B,\delta,a}$, whose set 
      of \strong{state} is equal to~$A$; the alphabet is equal to~$B$;
      the initial state is the letter~$a$, all states are accepting;
      and the transition function is defined by:
      $$ b\pathx{x} c \ee \text{if there exists~$i$ such that }\left\lbrace\begin{array}{l}\text{$c$ is the $i$-th letter of }\sigma(b)\\ \text{$x$ is the~$i$-th letter of }g(b)\end{array}\right.$$
    (\cf \rtable{sumup} for a summary of this transformation).
    \smallskip
    
    Note that there is loop on the initial state~$a$,
      since the morphism~$f$ is prolongable on a;
      we denote by~$x$ the label of this loop, that is, the smallest letter of~$g(a)$.
    This loops corresponds to the usual 0-loop differentiating \itrees from 
      trees and we will consider
      in the following the language~$L = L(\Ac)\xmd \bigcap \xmd ((B \setminus x). B^*)$
      where the loop is removed \emph{on the root only}.

    
    \smallskip
    
    Let us prove that the labelled signature of~$L$ is equal to~$(\signature,\labelling)$.
    We denote by~$(u_i)_i$ the following sequence of words (over $A$)~$u_0= a$,~~~$u_1=a^{-1}\sigma(a)$
      and for all~~~$i>0$,~$u_{i+1}=f(u_i)$.
    It follows that~$u_0\xmd u_1 \xmd \cdots \xmd u_i =  \sigma^i(a)$.
    
    Let us fix an~$i$ and consider the words of~$L$ of length~$i$ in 
      radix order; we denote by~$w_k$ the~$(k+1)$-th word of L of length~$i$.
    It can be easily proven (by induction over~$i$) that 
      $$ \forall k\in\N \ee a\pathx{w_k}[\Ac] c_k \ee \text{ where }c_k\text{ is the $k$-th letter of~$u_i$}\eqpnt$$
    
    It follows that the~$(k+1)$-th word (of any length) of~$L$ in radix order reaches in
      the automaton~$\Ac$,
      the~$k$-th letter of~$\sigma^\omega(a)$.
    Since the outgoing transitions of a state~$b\in A$ are labelled by the letters 
      of the words~$g(b)\in B^*$,
      it follows that the labelled signature of~$L$ is equal 
      to~$(\signature,\labelling)$.
%
%
%
  \end{proof}
  
  \begin{proposition}\lproposition{rat->subst}
    The labelled signature of a prefix-closed rational language is a substitutive signature.
  \end{proposition}
  \begin{proof} 
    Let~$L$ be a rational language and a finite minimal deterministic 
      trim automaton~$\Ac=\aut{\Sigma,Q,\delta,i}$ accepting the 
      language~\xmd$\#^*L$,~`$\#$' being a letter which does not appear in~$L$,
      and is fixed as smaller than every other letter.
    Reusing the transformation summed up in~\rtable{sumup}, we define the 
      morphisms:
      $$\begin{array}{lccl}
        \sigma:&\e Q^*\e&\longrightarrow \e & Q^* \\
               & p& \longmapsto \e & q_0 \xmd q_1 \xmd q_2\cdots q_k \\
      \end{array}
      \eee
      \begin{array}{lccl}
        g:&\e Q^*\e&\longrightarrow \e & \Sigma^* \\
               & p& \longmapsto \e & a_0\xmd a_1\xmd a_2\cdots a_k \\
      \end{array}$$
    where~$a_0<a_1<\cdots<a_k$, and for all~$i\in\intint{0}{k}$,~$p \pathx{a_i}[\Ac] q_i$.
    It follows from definition that~$\sigma$ is prolongable on the letter~$i\in Q$, 
      since it corresponds to the initial state of~$\Ac$
      on which there is necessarily a loop labelled by~$\#$.
    
    From there, the proof is essentially the same as the one from \rproposition{subst->rat}.
  \end{proof}

%
\section{Conclusion and future work}
In this work, we  introduced a way of effectively describing infinite 
  trees and languages
  by infinite words using a simple breadth-first traversal.
Since this transformation is essentially one-to-one, 
  it is natural to wonder which class of words is associated
  with which class of languages and which properties of the former
  can be translated into properties of the latter.
  
In this first work on the subject, we have proved that rational languages 
  are associated with (a particular subclass of) substitutive words.

In a forthcoming paper \cite{MarsSaka2014a}, we study the class of languages 
  associated with periodic signatures and how they are related to the 
  representation language in rational base numeration systems.
In both cases these results express the intimate relationship between signatures and numeration systems.

Our aim is to further explore this relationship 
  by means of the notion of direction of a signature, 
  that generalises the notion of growth ratio given at \rsection{subst}.
For instance, a rational base numeration system has a signature 
  deduced from the Christoffel word associated with that rational number.
  
%
%

\clearpage
\hfill
\bibliographystyle{plain}
\bibliography{%
  ../bibliography,%
  \BIBINPUTSDIR Alexandrie-abbrevs,%
  Alexandrie-AC,%
  Alexandrie-DF,%
  Alexandrie-GL,%
  Alexandrie-MR,%
  Alexandrie-SZ%
}
\addcontentsline{toc}{section}{References}
\clearpage
\appendix
\renewcommand{\thesection}{A.\arabic{section}}
\renewcommand{\thefigure}{A.\arabic{figure}}


\setcounter{section}{2}
\setcounter{thm}{0}

\begin{center}
  {\Large\bf Appendix: Some More Figures}
\end{center}

\vspace*{0.5cm}

\begin{figure}[h!]
  \centering
  \includegraphics[width=\linewidth]{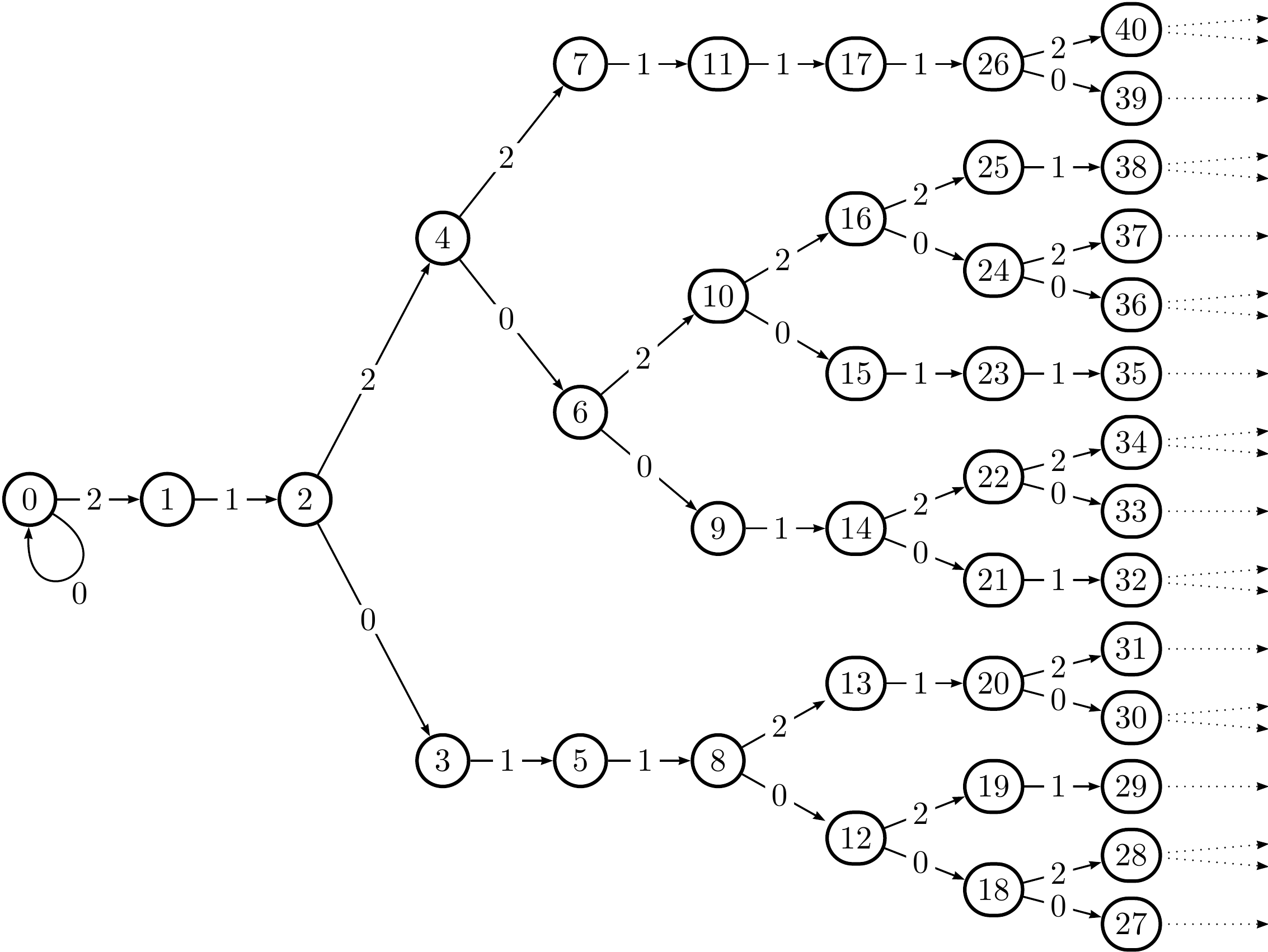}
  \caption{The language of the representations of integers in base~$\frac{3}{2}$, its signature is~$(21)^\omega$
  and its labelling is~$(021)^\omega$.}
  \lfigure{l32}
\end{figure}

\vspace*{0.5cm}

\begin{figure}[h!]
  \centering
  \includegraphics[scale=1.25]{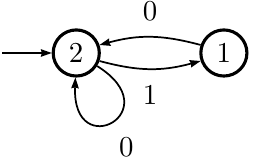}
  \caption{The automaton accepting the representations of integers in base Fibonnaci}
  \lfigure{fibo_m}
\end{figure}

%
%
%
%
\end{document}